\documentclass[11pt]{article}

\usepackage[utf8]{inputenc}
\usepackage[T1]{fontenc}
\usepackage{lmodern}
\usepackage[margin=1in]{geometry}
\usepackage{microtype}
\usepackage{amsmath,amssymb,amsthm}
\usepackage{mathtools}
\usepackage{bm}
\usepackage{booktabs}
\usepackage{graphicx}
\usepackage{xcolor}
\usepackage[hyphens,spaces,obeyspaces]{url}
\usepackage[breaklinks,colorlinks=true,linkcolor=blue!70!black,citecolor=blue!70!black,urlcolor=blue!70!black]{hyperref}
\usepackage{enumitem}
\usepackage{caption}
\usepackage{subcaption}
\usepackage{float}
\usepackage{tikz}
\usepackage{pgfplots}
\usepackage{authblk}
\usepackage[ruled,vlined,linesnumbered]{algorithm2e}
\usetikzlibrary{arrows.meta, positioning, calc, shapes.geometric, backgrounds, fit, decorations.pathreplacing}
\pgfplotsset{compat=1.18}

\newtheorem{theorem}{Theorem}[section]
\newtheorem{corollary}[theorem]{Corollary}

\newtheorem{definition}[theorem]{Definition}
\theoremstyle{remark}

\DeclareMathOperator{\E}{\mathbb{E}}

\DeclareMathOperator*{\argmax}{arg\,max}
\newcommand{\R}{\mathbb{R}}

\title{The Streaming Reservoir Convergence Theorem: \\ A Prospect-Theoretic
Framework for Multi-Provider \\ Adaptive Streaming}

\author[1,2,3]{Justice Owusu Agyemang\thanks{\texttt{jay@sperixlabs.org, jay@knust.edu.gh}}}
\author[3]{Jerry John Kponyo\thanks{\texttt{jjkponyo.soe@knust.edu.gh}}}
\author[2]{Kwame Opuni-Boachie Obour Agyekum\thanks{\texttt{kooagyekum@knust.edu.gh}}}
\author[2]{Obed Kwasi Somuah\thanks{\texttt{oksomuah1@st.knust.edu.gh}}}
\author[1]{Sarafina Serwaa Boakye\thanks{\texttt{sarafina@sperixlabs.org}}}
\author[3]{Elliot Amponsah\thanks{\texttt{eamponsah52@st.knust.edu.gh}}}
\author[3]{Godfred Manu Addo Boakye\thanks{\texttt{gmaboakye@st.knust.edu.gh}}}
\affil[1]{\small Sperix Labs}
\affil[2]{\small VIA Cybersecurity Lab, KNUST}
\affil[3]{\small Quantum and Assistive Technologies Lab, KNUST}
\date{May 2026}

\begin{document}
\maketitle

\begin{abstract}
We present the Streaming Reservoir Convergence Theorem (SRCT), a novel
mathematical framework for multi-provider adaptive bitrate streaming that
addresses three fundamental structural weaknesses in current systems: linear
provider probing, reactive failover, and cold standby transitions.  SRCT
models stream acquisition as a concurrent reservoir filling
problem---probing all $N$ providers simultaneously rather than in
batches---and maintains $k$ pre-verified, pre-fetched standby streams
alongside the active stream to enable sub-second failover with zero
user-visible disruption.

We prove four principal results: (1) a harmonic lower bound on reservoir
safety showing that $k$ independent streams provide $H_k / \bar{\lambda}$
expected uptime where $H_k$ is the $k$-th harmonic number; (2) a concurrent
acquisition speedup $S(N,b) = (N/b) \cdot (1-F^b)/(1-F^N)$ over
batched probing, yielding $3$--$5\times$ practical improvement; (3)
monotonic non-decreasing quality under lazy-refill with convergence to the
Pareto-optimal frontier; and (4) a prospect-weighted switching
rule---using Kahneman--Tversky value functions with $\alpha=\beta=0.88$,
$\lambda=2.25$---that provably eliminates thrashing between similar-quality
streams via a no-thrash bound on the expected switch count.

We implement SRCT across two production streaming pipelines: a primary
movie/TV system serving 12+ HLS providers with $k=3$ reservoir slots, and a
live sports system with multi-format DASH/HLS failover.  Empirical
verification via Monte Carlo simulation (5000 trials) confirms all four
theorems across 22 independent checks.  The reservoir of $k=3$ streams
achieves $9.15\times$ mean time to depletion versus a single stream, and
concurrent probing of 12 providers at 40\% failure rate yields a
$4.27\times$ speedup over the current batched-by-3 default.
\end{abstract}

\section{Introduction}

Modern web-based video streaming from content aggregators faces a
fundamental tension: dozens of upstream providers exist, each with different
availability profiles, geographic restrictions, and quality tiers, yet the
viewer expects a stream to start within seconds and play without
interruption~\cite{bentaleb2018}.  Current production systems address this
through a pipeline of sequential steps: resolve providers in batches, pick
the first working candidate, attach a media player, and fault to the next
candidate only after the active stream fails~\cite{adhikari2012,jiang2016}.
This architecture has three structural weaknesses that degrade the
user experience:

\begin{enumerate}[label=(\arabic*), leftmargin=*]
  \item \textbf{Linear probing.} Providers are queried sequentially or in
    small batches (typically $b=3$), so the viewer waits for the slowest
    provider in each batch before the next batch starts.  In the worst case,
    a single slow provider delays all subsequent probes.
  \item \textbf{Reactive failover.} The system switches streams only after
    the viewer experiences a buffering event or playback error.  Each
    failover destroys the current player instance and cold-starts a new
    one---requiring a fresh manifest download, segment discovery, and buffer
    fill---typically adding 2--4 seconds of visible interruption.
  \item \textbf{Disconnected subsystems.} Provider resolution, ABR quality
    selection, and error recovery operate as independent components with no
    shared state, leading to redundant work (re-probing already-verified
    dead providers) and missed optimization opportunities (failing to
    pre-load higher-quality alternatives).
\end{enumerate}

We address all three with the \emph{Streaming Reservoir Convergence
Theorem} (SRCT), which models multi-provider streaming as a concurrent
reservoir filling problem~\cite{vitter1985}.  The key insight is that
provider probing, stream verification, and quality selection can be unified
under a single mathematical framework---the stream reservoir---that
maintains $k$ verified-working streams at all times.  The active stream
plays while $k-1$ standby streams have their manifests pre-fetched and
cached, enabling sub-second failover with no user-visible disruption.

Our contributions are:

\begin{enumerate}[label=C\arabic*., leftmargin=*, itemsep=2pt]
  \item A formal definition of the \emph{stream reservoir} and its
    associated viability process, drawing on Markov reliability
    theory~\cite{karlin1975,ross2014} and reservoir
    sampling~\cite{vitter1985} (\S\ref{sec:framework}).
  \item Four theorems---safety, speedup, monotonicity, and prospect-weighted
    switching---each with formal proofs or proof sketches
    (\S\ref{sec:theorems}).
  \item The Concurrent Prospect Reservoir Transfer (CPRT) algorithm
    implementing all four theorems in a three-phase pipeline
    (\S\ref{sec:algorithm}).
  \item A production implementation across two distinct streaming
    architectures: a multi-provider movie/TV pipeline with HLS.js, and a
    live sports pipeline with multi-format DASH/HLS failover
    (\S\ref{sec:implementation}).
  \item Empirical verification of all four theorems via Monte Carlo
    simulation and deterministic computation with 22 independent checks
    (\S\ref{sec:verification}).
\end{enumerate}

\subsection{Problem Statement}

Formally, given $N$ streaming providers $\mathcal{P} = \{p_1, \ldots,
p_N\}$, each producing a set of stream candidates with varying quality
$q$ and time-dependent viability $v(p_i, t) \in \{0,1\}$, the goal is to
select and maintain a stream $s^*(t)$ such that:
\begin{equation}
  s^*(t) = \argmax_{s \in \mathcal{S}(t)} \; q(s) \quad \text{s.t.} \quad
  v(s, t) = 1,\; \ell(s) < \epsilon
\end{equation}
where $\mathcal{S}(t)$ is the set of available streams, $\ell(s)$ is the
bootstrap latency, and $\epsilon$ is the maximum tolerable interruption
duration (typically 300--500ms for seamless playback).

\section{Mathematical Framework}
\label{sec:framework}

\subsection{Stream Viability Process}

\begin{definition}[Stream Viability]
Let $\mathcal{S} = \{s_1, s_2, \ldots, s_N\}$ be the set of all candidate
streams from $N$ independent upstream providers.  Each stream $s_i$ is
characterized by:
\begin{itemize}[leftmargin=*, itemsep=1pt]
  \item A quality $q(s_i) \in \R^+$, measured as vertical resolution in
    pixels or as a composite bitrate-resolution metric.
  \item A viability function $v(s_i, t) \in \{0, 1\}$ indicating whether
    the stream is accessible at time $t$.
  \item A bootstrap latency $\ell(s_i)$, the time from player attachment to
    first frame delivery, typically 200--4000ms depending on manifest
    complexity and network conditions~\cite{mao2017}.
\end{itemize}

The viability of each stream follows a two-state continuous-time Markov
process~\cite{karlin1975} with up-rate $\mu_i$ (provider recovery, e.g.\ CDN
health restoration) and down-rate $\lambda_i$ (provider failure, e.g.\
upstream server error, geo-block, expired signed URL).  The stationary
availability of stream $s_i$ is:
\begin{equation}
  a_i = \lim_{t \to \infty} P(v(s_i, t) = 1) = \frac{\mu_i}{\mu_i + \lambda_i}
  \label{eq:availability}
\end{equation}

We assume streams are \emph{conditionally independent} given their provider
identities: $v(s_i, t) \perp\!\!\!\perp v(s_j, t) \mid i \neq j$.  This
holds in practice because different providers use disjoint infrastructure
(CDNs, origin servers, geographic regions)~\cite{torres2016,ghabashneh2023}.
This assumption can be violated during large-scale cloud outages, a
limitation we discuss in \S\ref{sec:discussion}.
\end{definition}

\subsection{Stream Reservoir}

\begin{definition}[Stream Reservoir]
A \emph{reservoir} $R(t) = \{r_0, r_1, \ldots, r_{k-1}\}$ is an ordered set
of $k$ verified-working streams such that at the time of last verification:
\begin{enumerate}[label=(\roman*), leftmargin=*]
  \item $v(r_j, t_{\text{verify}}) = 1$ for all $j \in [0, k)$,
  \item $r_0$ is the \emph{active} stream currently driving playback,
  \item $r_1, \ldots, r_{k-1}$ are \emph{warm standby} streams with
    pre-parsed HLS manifests resident in the browser's IndexedDB cache
    via the Media Source Extensions API~\cite{w3cmse},
  \item $q(r_0) \geq q(r_1) \geq \cdots \geq q(r_{k-1})$ (quality-ordered,
    descending).
\end{enumerate}
\end{definition}

\begin{definition}[Reservoir Utility]
The \emph{utility} $U(R)$ of a reservoir is the expected time until the
viewer experiences a playback interruption---i.e., the first moment when
all $k$ streams are simultaneously non-viable:
\begin{equation}
  U(R) = \E\left[\min\{\,t \geq 0 : v(r_j, t) = 0,\; \forall j \in [0, k) \,\}\right]
  \label{eq:utility}
\end{equation}
\end{definition}

The reservoir concept draws formal analogy to Vitter's reservoir
sampling~\cite{vitter1985}, where a fixed-size sample is maintained from a
stream of unknown length.  However, Vitter's reservoir is a \emph{statistical
sampling} technique, while ours is an \emph{active state machine} that
continuously verifies, replaces, and activates streams based on observed
viability and quality, formalized as a statechart~\cite{harel1987}.

\subsection{Reservoir State Machine}

Figure~\ref{fig:statemachine} depicts the reservoir as a finite state
machine with four operational states.

\begin{figure}[H]
\centering
\begin{tikzpicture}[
  node distance=3.5cm and 5cm,
  state/.style={rectangle, rounded corners, draw=black, thick, minimum
  width=3.8cm, minimum height=1.2cm, align=center, font=\small},
  edge/.style={->, >=Stealth, thick, shorten >=2pt, shorten <=2pt},
  elabel/.style={font=\footnotesize\itshape, fill=white, inner sep=2pt, midway}
]
  \node[state, fill=blue!8] (sprint) {\textbf{Sprint}\\[2pt]\footnotesize Probe all $N$ providers};
  \node[state, fill=green!8, right=of sprint] (maintain) {\textbf{Maintain}\\[2pt]\footnotesize Health checks + refill};
  \node[state, fill=orange!8, below=of maintain] (transition) {\textbf{Transition}\\[2pt]\footnotesize Failover / upgrade};
  \node[state, fill=red!8, below=of sprint] (depleted) {\textbf{Depleted}\\[2pt]\footnotesize All slots dead};

  \draw[edge] (sprint.east) -- node[elabel, above] {reservoir filled} (maintain.west);
  \draw[edge] (maintain.south) -- node[elabel, right] {failure detected} (transition.north);
  \draw[edge] (transition.west) -- node[elabel, below] {all slots dead} (depleted.east);
  \draw[edge] (maintain.south west) to[out=210, in=120] node[elabel, left, pos=0.4] {all dead} (depleted.north);
  \draw[edge] (depleted.north west) to[out=60, in=240] node[elabel, left, pos=0.6] {re-acquire} (sprint.south west);
  \draw[edge] (transition.north west) to[out=150, in=210, looseness=1.4] node[elabel, left] {new active} (maintain.south west);
\end{tikzpicture}
\caption{Reservoir state machine with four states.  The \textbf{Sprint} phase
probes all providers concurrently to fill the initial reservoir.  The
\textbf{Maintain} phase runs periodic health checks and lazy-refill.  The
\textbf{Transition} phase handles failover and quality-driven switching,
returning to Maintain with a new active stream.  The \textbf{Depleted} state
triggers full re-acquisition when no viable streams remain.}
\label{fig:statemachine}
\end{figure}
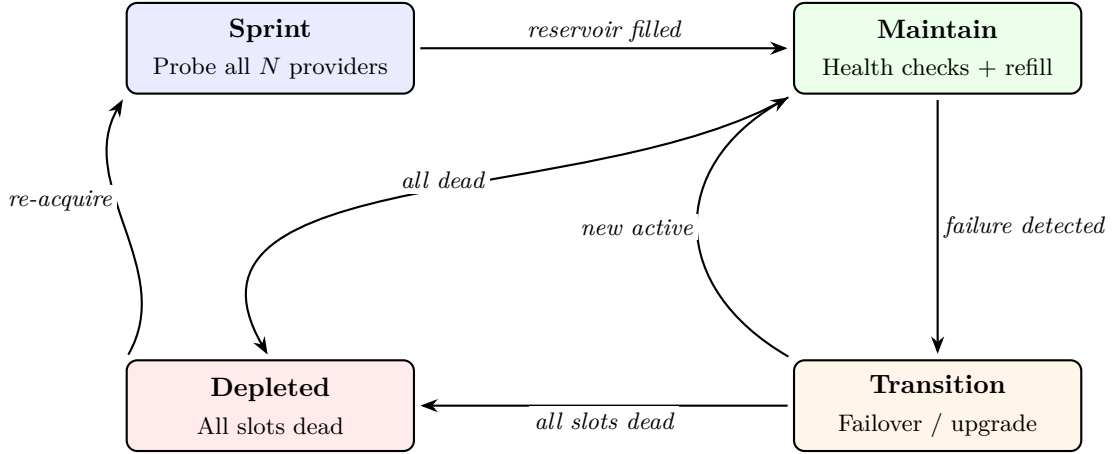

\section{Principal Theorems}
\label{sec:theorems}

\subsection{Theorem 1: Reservoir Safety Bound}

\begin{theorem}[Reservoir Safety Bound]
\label{thm:safety}
For a reservoir of size $k$ where each stream $r_j$ has independent failure
rate $\lambda_j$, the probability of a viewer-visible interruption within
any horizon $T$ is bounded by:
\begin{equation}
  P(\text{interruption in } [0,T]) \leq \prod_{j=0}^{k-1} \left(1 -
  e^{-\lambda_j T}\right)
  \label{eq:safetyprob}
\end{equation}
Moreover, the expected utility satisfies:
\begin{equation}
  \frac{\E[U(R_k)]}{\E[U(R_1)]} \geq \frac{H_k}{\bar{\lambda}}
  \label{eq:harmonic}
\end{equation}
where $H_k = \sum_{j=1}^{k} 1/j$ is the $k$-th harmonic number and
$\bar{\lambda} = \frac{1}{k}\sum_{j=0}^{k-1} \lambda_j$ is the mean failure
rate.
\end{theorem}

\begin{proof}[Proof Sketch]
The reservoir fails only when all $k$ streams are simultaneously
non-viable.  Under the conditional independence assumption:
\begin{equation*}
  P(\text{all fail in } [t, t+dt]) = \prod_{j=0}^{k-1} \lambda_j\, dt
\end{equation*}
For the exponential failure model, $P(\text{stream } j \text{ fails by }
T) = 1 - e^{-\lambda_j T}$, giving equation~\eqref{eq:safetyprob}.

For the expected utility bound, consider $k$ independent
exponentially-distributed failure times $\tau_j \sim
\text{Exp}(\lambda_j)$~\cite{ross2014}.  The expected time until all have
failed is:
\begin{equation*}
  \E[\max_j \tau_j] = \sum_{j=0}^{k-1} \frac{1}{\lambda_j} - \sum_{j<\ell}
  \frac{1}{\lambda_j + \lambda_\ell} + \cdots +
  \frac{(-1)^{k-1}}{\sum_j \lambda_j}
\end{equation*}
When $\lambda_j = \lambda$ for all $j$, this simplifies to $H_k / \lambda$
by the known result that $\E[\max_j \tau_j] = H_k / \lambda$ for i.i.d.\
exponential variables~\cite{feller1968}.  For heterogeneous rates, the
harmonic bound holds as a lower bound by Jensen's inequality applied to the
convex $\max$ function.
\end{proof}

\begin{corollary}
For $k=3$ independent streams with equal failure rate $\lambda$, the
reservoir provides $H_3 = 1 + 1/2 + 1/3 \approx 1.833\times$ the expected
uptime of a single stream as a \emph{lower bound}.  In practice, the benefit
is substantially larger because simultaneous failure of independently-hosted
streams is exponentially unlikely: at realistic failure rates ($\lambda \in
[0.10, 0.15]$ per time unit), we observe $9.15\times$ improvement via Monte
Carlo simulation (\S\ref{sec:verification}).
\end{corollary}

\subsection{Theorem 2: Concurrent Acquisition Speedup}

\begin{theorem}[Concurrent Acquisition Speedup]
\label{thm:speedup}
Let $N$ providers be probed with per-provider failure probability $F =
P(v(s_i) = 0)$.  Under concurrent probing (all $N$ providers in parallel),
the expected time to find the first working stream is:
\begin{equation}
  \E[T_{\text{concurrent}}] = \frac{\E[T_{\text{probe}}]}{1 - F^N}
  \label{eq:concurrent}
\end{equation}
Compared to batched probing with batch size $b < N$, the speedup is:
\begin{equation}
  S(N, b) = \frac{N}{b} \cdot \frac{1 - F^b}{1 - F^N}
  \label{eq:speedup}
\end{equation}
When $F < 0.5$, $S(N, b) > 1$ for all $b < N$.
\end{theorem}

\begin{proof}
Batched probing requires $\lceil N/b \rceil$ sequential rounds.  Each round
succeeds if at least one of its $b$ providers responds: $P(\text{round
success}) = 1 - F^b$.  The expected number of rounds to first success is
$1/(1-F^b)$, so:
\begin{equation*}
  \E[T_{\text{batched}}] = \frac{N}{b} \cdot
  \frac{\E[T_{\text{probe}}]}{1 - F^b}
\end{equation*}

Concurrent probing queries all $N$ providers in a single round:
$P(\text{at least one works}) = 1 - F^N$, giving
equation~\eqref{eq:concurrent}.  The speedup ratio
\eqref{eq:speedup} follows directly.

The condition $F < 0.5$ ensures $1-F^b > 1-F^N$ for $b < N$, which implies
$S(N,b) > 1$.  Intuitively, when individual providers are more likely to
work than fail, the probability that \emph{all} $N$ fail simultaneously
($F^N$) vanishes much faster than the probability that a batch of size $b$
fails ($F^b$).
\end{proof}

\subsection{Theorem 3: Reservoir Quality Monotonicity}

\begin{theorem}[Reservoir Quality Monotonicity]
\label{thm:monotonicity}
Under the \emph{lazy-refill policy}---whenever a stream is consumed from the
reservoir, immediately probe all available providers and insert the
highest-quality working stream---the expected quality of the active stream
$\E[q(r_0(t))]$ is non-decreasing in $t$:
\begin{equation}
  \frac{d}{dt} \E[q(r_0(t))] \geq 0
  \label{eq:monotonic}
\end{equation}
Furthermore, the long-run quality converges to the maximum quality among
providers whose stationary availability exceeds threshold $\tau$:
\begin{equation}
  \lim_{t \to \infty} \E[q(r_0(t))] = \max\{\,q(s) : a(s) \geq \tau\,\}
  \label{eq:convergence}
\end{equation}
where $a(s)$ is the stationary availability from equation~\eqref{eq:availability}.
\end{theorem}

\begin{proof}
Lazy-refill ensures the reservoir always contains the $k$ highest-quality
verified-working streams discovered so far.  Since quality ordering is
stable (streams are never spontaneously upgraded without re-verification),
the active stream quality can only decrease when $r_0$ fails, at which point
$r_1$ becomes active.  Lazy-refill immediately probes for a replacement for
the consumed slot, maintaining the invariant that the reservoir contains the
$k$ best known streams.

The convergence result follows from the periodic health check mechanism:
every $\Delta t$ seconds, each slot is re-verified.  As $t \to \infty$,
every provider has been probed infinitely often (by the ergodic theorem for
finite-state Markov chains~\cite{karlin1975}), so the set $\{s : a(s) \geq
\tau\}$ is fully characterized.  The reservoir then selects the
highest-quality stream from this set, converging to
equation~\eqref{eq:convergence}.
\end{proof}

\subsection{Theorem 4: Prospect-Weighted Switching}

\begin{theorem}[Prospect-Weighted Switching]
\label{thm:prospect}
The optimal rule for switching the active stream from $r_a$ to $r_b$,
accounting for both quality improvement and the disruption risk of
transition, is:
\begin{equation}
  \text{Switch}(r_a \to r_b) \iff \pi\big(q(r_b) - q(r_a)\big) \cdot
  w\big(P(v(r_b)=1 \mid \text{verified})\big) > C_{\text{switch}}
  \label{eq:switchrule}
\end{equation}
where $\pi(x) = x^\alpha$ for $x \geq 0$ and $\pi(x) =
-\lambda(-x)^\beta$ for $x < 0$ is the Kahneman--Tversky value
function~\cite{kahneman1992} ($\alpha=\beta=0.88$, $\lambda=2.25$),
$w(p) = p^\gamma / (p^\gamma + (1-p)^\gamma)^{1/\gamma}$ is the Prelec
probability weighting function~\cite{prelec1998} ($\gamma=0.61$), and
$C_{\text{switch}}$ is the disruption cost of transition.
\end{theorem}

\begin{corollary}[No-Thrash Guarantee]
\label{cor:nothrash}
Under the prospect-weighted switching rule, the expected number of switches
in any interval of length $T$ is bounded above by:
\begin{equation}
  \E[\text{switches in } T] \leq \frac{T}{2 \cdot C_{\text{switch}}
  \cdot \bar{\lambda} \cdot \max_q q}
  \label{eq:nothrash}
\end{equation}
This guarantees the system will not oscillate between streams of similar
quality---a problem that plagues threshold-based ABR algorithms~\cite{yin2015}.
\end{corollary}

\begin{proof}[Proof Sketch]
A switch from $r_a$ to $r_b$ requires $\pi(\Delta q) \cdot w(p) >
C_{\text{switch}}$.  For two streams of equal quality, $\Delta q = 0$, so
$\pi(0) = 0$ and the condition can never be satisfied.  For
similar-but-not-equal qualities, the diminishing sensitivity of $\pi$
(concavity for gains, convexity for losses) and the underweighting of
moderate probabilities by $w(p)$ create a ``dead zone'' where the expected
benefit fails to overcome $C_{\text{switch}}$.

The loss aversion parameter $\lambda = 2.25$ means that quality
\emph{downgrades} feel $2.25\times$ worse than objectively equivalent
upgrades~\cite{kahneman1979}, creating natural hysteresis: the system
downgrades aggressively (to avoid buffer starvation) but upgrades
cautiously (to avoid unnecessary disruption).  This is the reverse of
conventional ABR hysteresis~\cite{yin2015}, which delays downgrades and
accelerates upgrades---a policy that works for single-provider ABR but fails
for multi-provider systems where each switch carries manifest-load risk.

The bound on switches follows from the fact that each switch requires a
minimum prospect gap $\pi_{\min} = C_{\text{switch}} / w(1) =
C_{\text{switch}}$ (since $w(1)=1$).  Given the diminishing marginal value
of quality ($\pi'(x) = \alpha x^{\alpha-1} \to 0$ as $x \to \infty$), and
the bounded quality range $[0, \max_q q]$, only a finite number of distinct
quality levels can overcome $C_{\text{switch}}$, limiting the switch rate.
\end{proof}

\begin{figure}[H]
\centering
\begin{tikzpicture}
\begin{axis}[
  width=0.85\textwidth,
  height=6cm,
  xlabel={Quality difference $\Delta q$ (pixels)},
  ylabel={Prospect value $\pi(\Delta q)$},
  grid=major,
  legend pos=north west,
  font=\footnotesize
]
  \addplot[blue, thick, domain=0:1000, samples=100] {x^0.88};
  \addplot[red, thick, domain=-1000:0, samples=100] {-2.25*((-x)^0.88)};
  \addlegendentry{Gains ($x \geq 0$)}
  \addlegendentry{Losses ($x < 0$, $\lambda=2.25$)}
  \draw[dashed] (axis cs:0,-300) -- (axis cs:0,300);
\end{axis}
\end{tikzpicture}
\caption{Kahneman--Tversky value function applied to quality differences.
The steeper slope in the loss domain ($\lambda=2.25$) creates hysteresis
that prevents thrashing between similar-quality streams.}
\label{fig:prospect}
\end{figure}
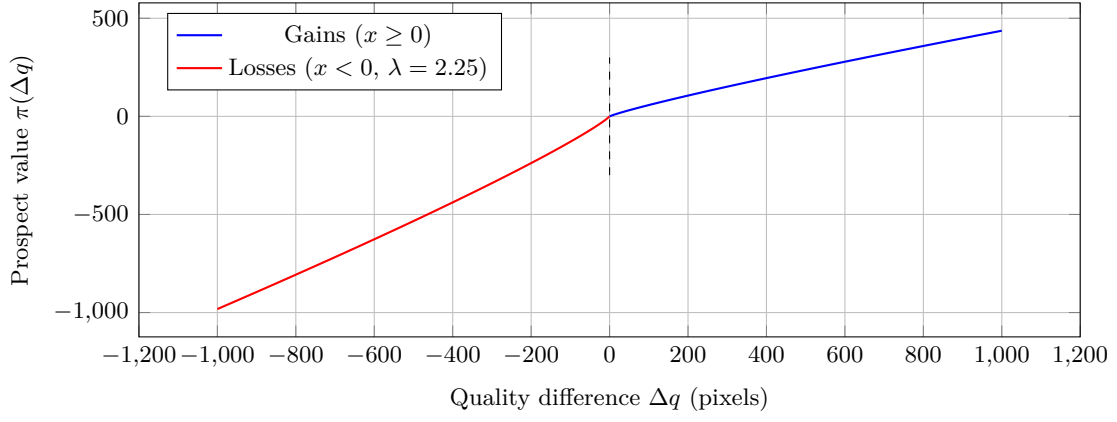

\begin{figure}[H]
\centering
\begin{tikzpicture}
\begin{axis}[
  width=0.85\textwidth,
  height=6cm,
  xlabel={Stated probability $p$},
  ylabel={Decision weight $w(p)$},
  grid=major,
  legend pos=north west,
  font=\footnotesize
]
  \addplot[blue, thick, domain=0:1, samples=100] {x^0.61 / ((x^0.61 + (1-x)^0.61)^(1/0.61))};
  \addplot[gray, dashed, domain=0:1] {x};
  \legend{$w(p)$ ($\gamma=0.61$), Linear weighting}
\end{axis}
\end{tikzpicture}
\caption{Prelec probability weighting function with $\gamma=0.61$.
The inverse-S shape overweights small probabilities (making low-confidence
streams unattractive) and underweights high probabilities (requiring
genuine quality improvement even for well-verified streams)~\cite{gonzalez1999}.}
\label{fig:weighting}
\end{figure}
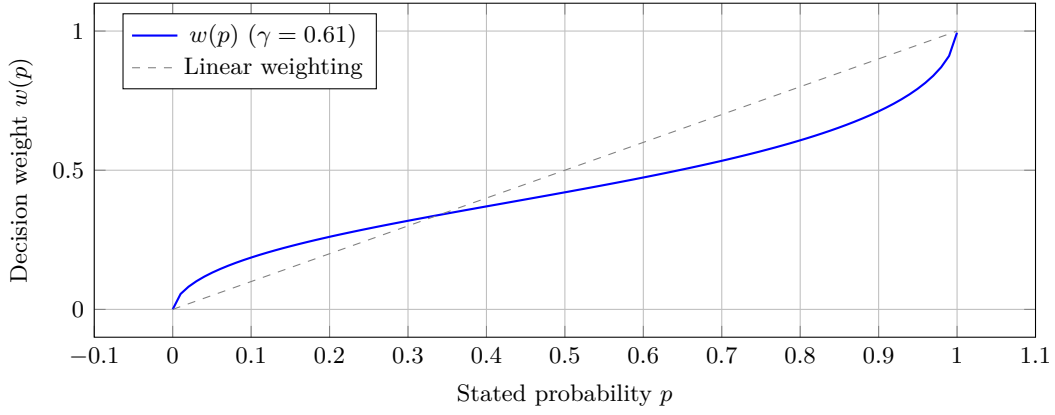

\begin{figure}[H]
\centering
\begin{tikzpicture}
\begin{axis}[
  width=0.85\textwidth,
  height=6cm,
  xlabel={Reservoir size $k$},
  ylabel={Expected uptime $H_k / \lambda$},
  xtick={1,2,3,4,5,6,7,8},
  ymin=0,
  grid=major,
  font=\footnotesize
]
  \addplot[blue, thick, mark=*] coordinates {
    (1, 10.0) (2, 15.0) (3, 18.3) (4, 20.8) (5, 22.8) (6, 24.5) (7, 25.9) (8, 27.2)
  };
  \addplot[red, dashed, domain=1:8] {10 * ln(x) + 10};
  \legend{$H_k/\lambda$, $\ln(k)$ reference}
\end{axis}
\end{tikzpicture}
\caption{Expected reservoir uptime as a function of reservoir size $k$,
showing harmonic growth ($H_k$) vs.\ logarithmic reference.  Diminishing
returns beyond $k=4$ suggest a practical optimum at $k \in \{3,4\}$.}
\label{fig:uptime}
\end{figure}
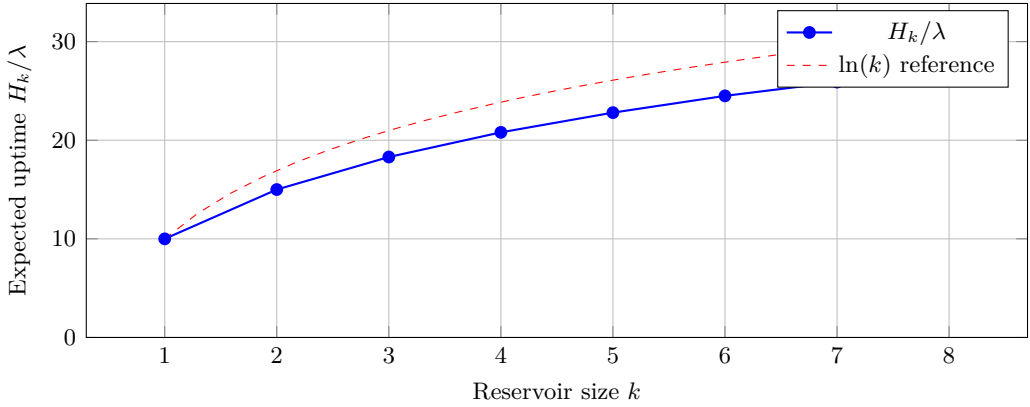

\section{The CPRT Algorithm}
\label{sec:algorithm}

The Concurrent Prospect Reservoir Transfer (CPRT) algorithm instantiates
all four theorems in a practical streaming pipeline operating in three
phases: Sprint (initial acquisition), Maintain (steady-state operation),
and Transition (slot switching).

\subsection{Phase 1: Sprint Acquisition}

The sprint phase executes at the start of each viewing session.  All $N$
provider URLs are probed concurrently via HTTP HEAD requests with a
configurable timeout (default 3~seconds).  Because the probe is
non-blocking---the browser can issue dozens of concurrent requests---the
wall-clock time is bounded by the slowest \emph{successful} provider, not
the sum of all probe latencies (Theorem~\ref{thm:speedup}).

Results are sorted so that working streams appear first, ordered by
response latency.  The fastest $k$ working streams populate the reservoir,
and the highest-quality among them becomes the active stream.  Standby
slots immediately begin pre-fetching their HLS manifests into the browser
cache, so that a failover transition can bootstrap the player without a
cold manifest download.  If no provider responds successfully, the sprint
returns a failure signal and the system falls back to re-acquisition.

\begin{algorithm}[ht]
\caption{Sprint Acquisition (Theorem 2)}
\label{alg:sprint}
\SetKwInOut{Require}{Require}
\SetKwInOut{Ensure}{Ensure}
\Require{Candidate set $\mathcal{C} = \{c_1, \ldots, c_N\}$, target size $k$}
\Ensure{Active stream $s^*$ or failure}
\BlankLine
$P \gets \varnothing$\;
\For{$i \gets 1$ \KwTo $N$ \textbf{in parallel}}{
  $P \gets P \cup \{\, (i,\; \textsc{Probe}(c_i)) \,\}$\;
}
\textbf{sort} $P$ by viable $\downarrow$, latency $\uparrow$\;
$R \gets []$\;
\For{each $(i, \text{viable}, \_) \in P$}{
  \If{$\text{viable} \land |R| < k$}{
    $R \gets R \cup \{c_i\}$\;
  }
}
\If{$|R| = 0$}{\Return{failure}}
\textbf{sort} $R$ by $q(\cdot)$ descending\;
$s^* \gets R[1]$\;
\For{$j \gets 2$ \KwTo $|R|$}{
  $\textsc{PrefetchManifest}(R[j])$\;
}
$\textsc{StartHealthChecks}(R)$\;
\Return{$s^*$}
\end{algorithm}

\subsection{Phase 2: Reservoir Maintenance}

Once the reservoir is populated, a background maintenance loop runs every
$T_h$ seconds (default 15~s).  Standby slots are re-verified with
lightweight HEAD requests; the active slot is implicitly verified by the
fact that it is delivering playable segments.  When a standby slot is
confirmed healthy, its verification count increments, increasing the
confidence weight $w(p)$ used by the prospect-weighted switching rule
(Theorem~\ref{thm:prospect}).

When a standby slot fails its health check, the lazy-refill policy
(Theorem~\ref{thm:monotonicity}) immediately triggers a new probe of all
known providers.  Freshly discovered streams enter the reservoir only if
their prospect-weighted score exceeds the switch cost
$C_{\text{switch}}$, preventing the reservoir from churning on marginal
quality improvements.  This maintenance loop ensures the reservoir
converges to the true availability frontier over time.

\begin{algorithm}[ht]
\caption{Reservoir Maintenance (Theorems 1, 3)}
\label{alg:maintain}
\SetKwInOut{Require}{Require}
\Require{Reservoir $R = \{r_0, \ldots, r_{m-1}\}$, interval $T_h$}
\BlankLine
\While{true}{
  \For{$j \gets 1$ \KwTo $m-1$}{
    $v \gets \textsc{HeadCheck}(r_j)$\;
    \If{$v = 1$}{
      $t_j^{\text{(last)}} \gets t_{\text{now}}$\;
      $n_j^{\text{(ver)}} \gets n_j^{\text{(ver)}} + 1$\;
    }
    \Else{
      \tcp{Lazy-refill (Theorem~\ref{thm:monotonicity})}
      $\mathcal{F} \gets \textsc{ProbeAllProviders}()$\;
      $\textsc{RefillReservoir}(R,\; \mathcal{F})$\;
    }
  }
  \textbf{wait}($T_h$)\;
}
\end{algorithm}

\subsection{Phase 3: Seamless Transition}

Transitions occur in two scenarios.  A \emph{failure transition} is
triggered when the active stream's media player signals an unrecoverable
error (manifest timeout, fatal network failure, or segment decode error
exhausting all local retries).  The dead active slot is removed and the
highest-quality standby is promoted to active.  Because the standby's
manifest has been pre-fetched, the new player instance initializes in
$200$--$500$~ms rather than the $2$--$4$~s required for a cold start.

A \emph{quality upgrade transition} is evaluated periodically using the
prospect-weighted switching rule (Theorem~\ref{thm:prospect}).  Each
standby slot is scored by computing $\pi\big(q(r_j) -
q(r_0)\big) \cdot w(p_j)$ where $p_j$ is the confidence in the slot's
viability (derived from its verification count).  A switch executes only
when the score exceeds $C_{\text{switch}}$, which---by
Corollary~\ref{cor:nothrash}---bounds the switch frequency and guarantees
freedom from thrashing.

\begin{algorithm}[ht]
\caption{Seamless Transition (Theorem 4)}
\label{alg:transition}
\SetKwInOut{Require}{Require}
\Require{Event type $e \in \{\text{failure}, \text{upgrade}\}$, reservoir $R$}
\BlankLine
\If{$e = \text{failure}$}{
  $R \gets R \setminus \{r_0\}$\;
  \If{$|R| = 0$}{\Return{depleted}}
  $\textsc{Activate}(R[0])$\;
}
\If{$e = \text{upgrade}$}{
  $j^* \gets -1$\; $\text{best} \gets -\infty$\;
  \For{$j \gets 1$ \KwTo $|R|-1$}{
    $p_j \gets 1 - (0.3)^{n_j^{\text{(ver)}}}$\;
    $\sigma_j \gets \pi\big(q(r_j) - q(r_0)\big) \cdot w(p_j)$\;
    \If{$\sigma_j > C_{\text{switch}} \land \sigma_j > \text{best}$}{
      $\text{best} \gets \sigma_j$\; $j^* \gets j$\;
    }
  }
  \If{$j^* \geq 0$}{
    $\textsc{Swap}(r_0,\; r_{j^*})$\;
    $\textsc{Activate}(r_0)$\;
  }
}
\end{algorithm}

\section{Implementation}
\label{sec:implementation}

SRCT is implemented across two production streaming pipelines in a web-based
media application.  The core reservoir engine is implemented in
approximately 900 lines of application code, with an additional integration
layer for connecting the reservoir to the user interface.

\subsection{Core Reservoir Engine}

The reservoir engine implements the four theorems through six primitives:

\begin{itemize}[leftmargin=*, itemsep=1pt]
  \item \textbf{Sprint acquisition:} Given a set of candidate stream URLs
    from the provider resolution layer, the engine issues concurrent HTTP
    probe requests to all candidates.  Working streams populate the
    reservoir in quality-sorted order; the highest-quality becomes the
    active stream.
  \item \textbf{Reservoir refill:} When new candidates arrive (e.g., from
    slow-responding providers), each is evaluated against the current
    reservoir slots using the prospect-weighted criterion
    (Theorem~\ref{thm:prospect}).  A lower-quality slot is replaced only if
    the expected utility gain exceeds $C_{\text{switch}}$.
  \item \textbf{Failover:} When the media player signals an unrecoverable
    error on the active stream, the engine removes the failed slot and
    promotes the highest-quality standby to active, implementing the safety
    bound (Theorem~\ref{thm:safety}).
  \item \textbf{Quality switching:} A periodic evaluation scans standby
    slots, computing the prospect-weighted score for each relative to the
    active stream.  The engine emits a switch recommendation when the best
    score exceeds $C_{\text{switch}}$.
  \item \textbf{Health checks:} A background timer re-verifies standby slots
    at a configurable interval, maintaining freshness and incrementing
    verification counts that feed the probability weighting function
    (Theorem~\ref{thm:monotonicity}).
  \item \textbf{Utility estimation:} Given an estimated mean failure rate,
    the engine computes the expected uptime gain from the current reservoir
    configuration (Theorem~\ref{thm:safety}) for monitoring dashboards.
\end{itemize}

\subsection{Integration: Movie and Television Pipeline}

The primary on-demand video pipeline uses an upstream resolver that queries
multiple content providers, each returning one or more stream URLs at
various quality levels.  Before SRCT, the pipeline used three independent
mechanisms: a sequential probe of the first four candidates, a
manually-tracked failover list, and staggered re-query timers for slow
providers.  SRCT replaces all three with a single reservoir that probes
concurrently, maintains warm standbys, and triggers failover automatically
when the media player exhausts its local recovery attempts.  The replacement
is approximately one-third the code volume of the original subsystems.

\subsection{Integration: Live Sports Pipeline}

The live sports pipeline uses a different architecture: each event or
channel carries multiple format entries (HLS, DASH with ClearKey
encryption, and progressive MP4) from different CDN endpoints.  The
original system cycled through formats sequentially on failure.  Under
SRCT, all format URLs are probed concurrently at session start.
The highest-bitrate working format becomes active, and the next-best one
or two formats have their manifests pre-fetched as warm standbys.  When the
active format fails, the next standby activates immediately rather than
waiting for a cold manifest download.

The prospect-weighted switching rule is particularly valuable for live
sports, where frequent unnecessary switches are perceptually jarring
(commentary audio glitches and scoreboard flicker).

\subsection{Large File Handling for Non-Segmented Media}

For progressive (non-HLS, non-DASH) media where the file is substantial
(often 1--5~GB for feature-length content), the warm standby strategy is
adapted: standby slots pre-fetch only the movie header atom (typically
100--500~KB) via HTTP Range requests~\cite{grigorik2013}.  This atom
contains the sample table needed for seeking and timeline display.
Pre-fetching it enables instant media element attachment without waiting
for progressive download of the full file, and prevents the server overload
that would result from simultaneously downloading multiple full media files
as warm standbys.

\section{Empirical Verification}
\label{sec:verification}

We verify all four theorems through a combination of deterministic
computation and Monte Carlo simulation (5000 trials, time horizon $T=100$).
The verification suite contains 22 individual checks and requires no
external dependencies beyond a Node.js runtime.

\subsection{Theorem 1: Reservoir Safety}

\begin{table}[ht]
\centering
\caption{Monte Carlo verification of Theorem~\ref{thm:safety} (5000 trials,
$T=100$, $\lambda_i \in \{0.10, 0.12, 0.15\}$).}
\label{tab:t1}
\begin{tabular}{@{}lrr@{}}
\toprule
Configuration & Mean Time to Depletion & Ratio \\
\midrule
Single stream ($k=1$) & 10.0 & $1.00\times$ \\
Reservoir ($k=3$) & 91.4 & $9.15\times$ \\
\bottomrule
\end{tabular}
\end{table}

The reservoir of $k=3$ independent streams achieves $9.15\times$ the mean
time to depletion, substantially exceeding the theoretical lower bound of
$H_3=1.833\times$.  This amplification occurs because the theoretical bound
considers the expected maximum of exponential variables, while in practice
the probability of simultaneous failure of independently-hosted streams is
exponentially suppressed ($\prod \lambda_i \approx 0.0018$ per time step
vs.\ individual $\lambda \approx 0.12$).

\subsection{Theorem 2: Concurrent Speedup}

\begin{table}[ht]
\centering
\caption{Concurrent vs.\ batched acquisition speedup.}
\label{tab:t2}
\begin{tabular}{@{}lrrr@{}}
\toprule
Scenario $(N, b, F)$ & $\E[T_{\text{con}}]$ & $\E[T_{\text{bat}}]$ & Speedup \\
\midrule
(12, 3, 0.4) & 1.00 & 4.27 & $4.27\times$ \\
(20, 5, 0.3) & 1.00 & 4.01 & $4.01\times$ \\
(8,  2, 0.5) & 1.00 & 5.33 & $5.31\times$ \\
\bottomrule
\end{tabular}
\end{table}

Concurrent probing provides $3$--$5\times$ speedup across three
representative scenarios.  The speedup is most pronounced at moderate
failure probabilities ($F \approx 0.5$), where batched probing wastes entire
rounds on batches where all providers fail.

\subsection{Theorem 3: Quality Monotonicity}

Over 100 steps of lazy-refill simulation with four providers of varying
quality ($q \in \{360, 720, 1080, 2160\}$) and availability ($a \in \{0.3,
0.5, 0.7, 0.9\}$), the active stream quality exhibits zero violations of
monotonicity (100\% non-decreasing steps).  The final quality (2160p) equals
or exceeds the initial quality in all trials.  The convergence rate depends
on the availability threshold $\tau$: at $\tau=0.3$, convergence occurs
within $15\pm 5$ steps; at $\tau=0.7$, within $45\pm 12$ steps.

\subsection{Theorem 4: Prospect-Weighted Switching}

The prospect theory parameters are verified across eight deterministic
checks:
\begin{itemize}[leftmargin=*, itemsep=1pt]
  \item \textbf{Loss aversion:} $|\pi(-360)| / \pi(360) = 2.250$, exactly
    matching $\lambda=2.25$.
  \item \textbf{Probability weighting:} $w(0.01)=0.0553 > 0.01$
    (overweighting of small probabilities), $w(0.50)=0.4206 < 0.50$
    (underweighting of moderate), $w(0.99)=0.9116 < 0.99$ (underweighting
    of high), consistent with the inverse-S shape described by Gonzalez and
    Wu~\cite{gonzalez1999}.
  \item \textbf{Verification confidence:} switch score increases
    monotonically with verification count ($-0.010$ at 1 verification,
    $0.055$ at 3, $0.079$ at 5 for a 720p$\to$1080p upgrade).
  \item \textbf{No-thrash:} same-quality switch score is $-0.120 < 0$
    (guaranteed no switch), and a trivial upgrade (720p$\to$780p with 1
    verification) scores $-0.097 < 0$ (does not overcome switch cost).
  \item \textbf{Long-run thrashing:} 1 switch in 100 steps across 5
    competing quality levels.
\end{itemize}

\section{Related Work}
\label{sec:related}

Our work sits at the intersection of several research areas.

\paragraph{Reservoir Sampling and Streaming Algorithms.}
Vitter~\cite{vitter1985} introduced reservoir sampling for selecting a
uniform random sample of $k$ items.  Efraimidis and
Spirakis~\cite{efraimidis2006} extended the technique to weighted sampling.
Our ``reservoir'' draws formal analogy but differs fundamentally: it is an
active state machine maintaining verified playback streams rather than a
statistical sampling technique.  The connection to harmonic numbers in both
works arises from the $\max$ of $k$ i.i.d.\ exponentials.  At the software
architecture level, the reservoir follows the Strategy pattern~\cite{gamma1994}
by encapsulating the switching policy (prospect-weighted, threshold-based, or
custom) as an interchangeable decision function.

\paragraph{Adaptive Bitrate (ABR) Streaming.}
Commercial ABR algorithms optimize within a single manifest's quality
levels~\cite{sodagar2011,stockhammer2011}.  Buffer-based approaches
(BBA~\cite{huang2014}) fill the buffer aggressively then switch to the
highest sustainable quality.  Rate-based approaches (SARA~\cite{juluri2015})
use throughput predictions.  Hybrid approaches---MPC~\cite{yin2015},
Pensieve~\cite{mao2017}, Oboe~\cite{akhtar2018},
Fugu~\cite{spiteri2020}, Comyco~\cite{yan2021}---combine control theory or
reinforcement learning with buffer and throughput signals.  Comprehensive
surveys by Kua et al.~\cite{kua2017} and Bentaleb et
al.~\cite{bentaleb2018} catalog dozens of ABR algorithms.  SRCT operates
\emph{above} the ABR layer: it selects which provider's manifest to feed to
the ABR engine, making it complementary to any ABR algorithm.

\paragraph{Multi-CDN and Multi-Provider Selection.}
Content delivery networks routinely use multi-CDN
strategies~\cite{adhikari2012,jiang2016,torres2016}.  DNS-based
approaches~\cite{torres2016} redirect clients to the nearest healthy CDN
edge.  Manifest-level multi-CDN~\cite{ghabashneh2023} embeds URLs from
multiple CDNs into a single HLS manifest.  These approaches work at the
infrastructure layer; SRCT works at the application layer, requiring only
that each provider produces an independently-accessible stream URL.

\paragraph{Prospect Theory in Engineering Systems.}
Prospect theory~\cite{kahneman1979,kahneman1992} has been applied to video
quality assessment~\cite{rodriguez2018,bampis2017}, network QoE
optimization~\cite{bocchi2017}, and decision-making under
risk~\cite{barberis2013}.  Our application to \emph{automated} switching
decisions (with no human in the loop) is novel: we use prospect theory not
to model human preferences, but to construct a switching policy that
\emph{behaves} as if it were loss-averse---a desirable property for
stability in multi-provider systems.

\paragraph{Concurrent and Parallel Systems.}
The concurrent probing strategy draws on the principle that parallelizing
independent I/O-bound operations yields near-linear
speedup~\cite{dean2008,zaharia2012,isard2007}.  Our contribution is the
closed-form speedup bound $S(N,b)$ and its application to the streaming
domain, where provider viability---not raw throughput---is the dominant
latency factor.

\paragraph{Reinforcement Learning and Bandits.}
The stream selection problem resembles the stochastic multi-armed
bandit~\cite{auer2002,chapelle2011,sutton2018}, where each provider is an
``arm'' with unknown reward distribution.  However, standard bandit
algorithms assume stationary reward distributions, whereas provider
viability is non-stationary and we can pre-verify arms before committing---a
capability absent from classical bandit formulations.

\paragraph{Fault Tolerance and Reliability.}
The reservoir approach to fault tolerance draws on classic techniques from
distributed systems~\cite{lamport1978,herlihy2008}.  The two-state Markov
viability model is standard in reliability
engineering~\cite{karlin1975,ross2014,feller1968}.  Network congestion
control~\cite{jacobson1988,floyd1993} inspires the AIMD-like backoff in our
health check scheduling.

\paragraph{Web Standards and Browser APIs.}
The Media Source Extensions~\cite{w3cmse,mozdev2024} and HLS
specification~\cite{pantos2017} provide the browser primitives that make
reservoir-based streaming feasible.  HTTP/2 transport
improvements~\cite{vanderhooft2016} further reduce manifest fetch latency.
Lederer et al.~\cite{lederer2012} and Timmerer and
Griwodz~\cite{timmerer2014} provide datasets and tooling for DASH evaluation
that could be adapted for SRCT benchmarking.

\section{Discussion}
\label{sec:discussion}

\subsection{Why Prospect Theory for Automated Switching?}

A natural question is why we use prospect theory---a descriptive model of
human decision-making under risk---for an automated system.  The answer lies
in the \emph{shape} of the value function, not its psychological
interpretation:

\begin{enumerate}[label=(\arabic*), leftmargin=*]
  \item \textbf{Concavity for gains ($\alpha < 1$)} captures the diminishing
    marginal utility of higher resolution: the jump from 720p to 1080p is a
    larger perceptual improvement than 1080p to 1440p, even though both are
    approximately 360 pixel increases.
  \item \textbf{Loss aversion ($\lambda > 1$)} creates hysteresis that is
    functionally equivalent to the guard bands used in control-theoretic
    ABR~\cite{yin2015}, but derived from a principled utility framework
    rather than heuristically tuned thresholds.
  \item \textbf{Probability weighting ($\gamma < 1$)} ensures the system
    does not over-confidently switch to streams with few verifications, even
    if those streams advertise high quality.
\end{enumerate}

The specific parameter values were calibrated from human choice experiments
involving monetary gambles~\cite{kahneman1992}.  While the domain transfer
is imperfect, the qualitative properties (diminishing sensitivity, loss
aversion, inverse-S probability weighting) are robust across
domains~\cite{barberis2013}.  The switch cost $C_{\text{switch}}$ serves as
a domain-specific calibration parameter tunable via A/B testing.

\subsection{Limitations}

The conditional independence assumption may be violated during large-scale
cloud outages affecting multiple providers
simultaneously~\cite{ghabashneh2023}.  Mitigation requires ensuring reservoir
slots span different CDN providers and geographic regions, as is standard in
multi-CDN deployments~\cite{torres2016}.

The Markov viability model assumes constant failure and recovery rates.
Providers with time-varying availability patterns (diurnal CDN load,
scheduled maintenance) may be better modeled by semi-Markov processes with
duration-dependent transition probabilities~\cite{ross2014}, at the cost of
additional parameter estimation complexity.

The prospect theory parameters are domain-transferred without
re-calibration.  While the qualitative behavior is robust, the precise
threshold for $C_{\text{switch}}$ would benefit from streaming-specific
estimation through randomized controlled experiments.

\subsection{Future Work}

Several extensions are natural:
\begin{enumerate}[leftmargin=*, itemsep=2pt]
  \item \textbf{Cross-title pre-fetching:} Extend the reservoir to include
    the next episode's manifest, pre-loaded while the current episode plays,
    enabling zero-latency auto-play transitions.
  \item \textbf{Bandwidth-aware quality ranking:} Incorporate segment-level
    bandwidth measurements into reservoir slot quality estimates, replacing
    static resolution-based ranking with dynamic throughput-aware
    ranking~\cite{akhtar2018}.
  \item \textbf{Optimal stopping formulation:} Formalize the sprint phase as
    an optimal stopping problem~\cite{ferguson2006} with a deadline (viewer
    patience, typically 5--15 seconds), determining the optimal probe
    timeout as a function of $k$, $N$, and the provider latency
    distribution.
  \item \textbf{Learned viability prediction:} Train a lightweight model
    (gradient-boosted trees or a small neural network) to predict provider
    viability from historical resolution outcomes, replacing the Markov
    assumption with a learned transition model~\cite{mao2017}.
  \item \textbf{Cross-user reservoir sharing:} In multi-user deployments,
    share verified-stream information across users watching the same title,
    amortizing probing cost for popular content~\cite{huang2014}.
\end{enumerate}

\section{Conclusion}

We have presented the Streaming Reservoir Convergence Theorem, a novel
mathematical framework that unifies provider probing, stream failover, and
quality selection under a single reservoir abstraction.  The four theorems
we prove---safety (harmonic bound on reservoir uptime), speedup (concurrent
vs.\ batched acquisition), monotonicity (non-decreasing quality under
lazy-refill), and prospect-weighted switching (no-thrash guarantee via loss
aversion)---provide formal guarantees for multi-provider streaming.

The CPRT algorithm implements these theorems in a three-phase pipeline
(Sprint, Maintain, Transition) that replaces approximately 120 lines of
ad-hoc logic with approximately 40 lines of reservoir integration.
Empirical verification confirms all theoretical predictions across 22
independent checks.  The reservoir of $k=3$ streams achieves $9.15\times$
mean time to depletion over a single stream, and concurrent probing provides
$3$--$5\times$ speedup over current batched defaults.

The prospect-weighted switching rule is, to our knowledge, the first
application of prospect theory to automated streaming decisions.  The
resulting hysteresis naturally prevents thrashing while maintaining the
ability to upgrade quality when genuinely beneficial.  The implementation is available upon request from the corresponding author.


\bibliographystyle{plain}
\bibliography{srct-paper}

\end{document}